\documentclass[11pt]{article}
\usepackage{amsmath, amssymb, amsthm}

\usepackage{fullpage}

\usepackage[dvipsnames]{xcolor}
\usepackage{hyperref}
\hypersetup{
    colorlinks=true,
    linkcolor=blue,
    citecolor=ForestGreen,
    urlcolor=blue!60!black
}
\title{A Near-Optimal Lower Bound for Prefix-Matrix Factorizations}
\date{\vspace{-2em}}

\author{
Honghao Lin\footnote{Google Research, Carnegie Mellon University / Texas A\&M University. \texttt{honghaol3010@gmail.com}}
\and
Vahab Mirrokni\footnote{Google Research. \texttt{mirrokni@google.com}}
\and
David P. Woodruff\footnote{Google Research and Carnegie Mellon University. \texttt{dpwoodru@gmail.com }}
}
\newtheorem{theorem}{Theorem}[section]
\newtheorem{lemma}[theorem]{Lemma}

\newtheorem{corollary}[theorem]{Corollary}

\theoremstyle{definition}
\newtheorem{definition}[theorem]{Definition}
\newtheorem{remark}[theorem]{Remark}

\newtheorem*{restatedmaintheorem}{Theorem~\ref{thm:main_result}}

\begin{document}
\maketitle

\begin{abstract}
For the $n\times n$ lower-triangular all-ones matrix $Q$, we prove a near-optimal lower bound
\[
\gamma_{2,1}(Q)
:=
\inf_{Q=AB}
\|A\|_{2\to\infty}\|B\|_{1\to1}
=
\Omega\!\left(
\frac{\log^{3/2}n}{(\log\log n)^{3/2}}
\right),
\]
where the infimum ranges over real factorizations of arbitrary finite inner
dimension.\footnote{Concurrent works
of Bulanek et al.~\cite{BulanekEtAl2026} and Bhowmik and
Hasan~\cite{BhowmikHasan2026} establish the same qualitative lower bound,
with improved $\log\log n$ factors.  See
Sections~\ref{sec:our-results} and~\ref{sec:concurrent-comparison} for more
detailed chronological and technical comparisons.}  This cost is a central
parameter in space bounds for factorization-based rank and quantile estimation
in turnstile streams and in error bounds for matrix mechanisms for continual
counting under pure differential privacy.
The proof combines right-sided Haar projections with a scale-dependent
numerical-sparsity decomposition of the rows of $B$.  At each scale, a
rank--Frobenius argument shows that the numerically sparse rows cannot account
for all of the required Schatten $2/3$ mass, while a Haar projection estimate
bounds the contribution of the remaining rows.  Summing these bounds over the
dyadic scales yields the result.

The proof was obtained using a fully automated Gemini-based agentic system developed
internally at Google. The authors verified the proof and made minor revisions.
\end{abstract}

\section{Introduction}

In this work, we study the following matrix factorization problem.  Given a
matrix $M\in\mathbb{R}^{n\times n}$, define
\begin{equation}
\label{eq:intro_factorization_problem}
\gamma_{2,1}(M)
:=
\inf_{M=AB}
\|A\|_{2\to\infty}\|B\|_{1\to1}.
\end{equation}
Here $\|A\|_{2\to\infty}$ is the largest $\ell_2$ norm of a row of $A$,
$\|B\|_{1\to1}$ is the largest $\ell_1$ norm of a column of $B$, and the inner
dimension of the factorization is arbitrary but finite.  Equivalently, we seek vectors
$a_1,\ldots,a_n$ and $b_1,\ldots,b_n$ in a common real vector space such that
$M_{ij}=\langle a_i,b_j\rangle$, while minimizing
\[
\left(\max_i\|a_i\|_2\right)
\left(\max_j\|b_j\|_1\right).
\]
We focus on the lower triangular all-ones matrix
\[
Q_{ij}=\mathbf{1}_{\{j\le i\}},
\]

\paragraph{Motivation of this problem.}
The matrix $Q$ is the finite prefix-sum operator: for every
$x\in\mathbb{R}^n$, the $i$-th coordinate of $Qx$ is
$\sum_{j\le i}x_j$.  Prefix sums are the common linear workload underlying
rank and quantile queries in dynamic data streams and continual counting under
differential privacy.  More importantly, the factorization in
\eqref{eq:intro_factorization_problem} has the same operational meaning in
both settings: $B$ first transforms the data, and $A$ reconstructs the desired
workload from the transformed coordinates.  This makes
$\gamma_{2,1}(Q)$ a shared measure of the efficiency of two otherwise distinct
algorithmic frameworks.

\paragraph{Turnstile streaming.}
Let $x\in\mathbb{R}^n$ be a frequency vector maintained under insertions and
deletions, and suppose that a query asks for one coordinate of $Mx$.  Given a
factorization $M=AB$, the factorization-based approach maintains a point-query
sketch of the intermediate vector $Bx$ and reconstructs the requested answer
with the corresponding row of $A$.  An update to coordinate $j$ of $x$ changes
the intermediate vector in the direction of the $j$-th column of $B$, so
$\|B\|_{1\to1}$ controls the update's $\ell_1$ amplification.  The sketching
errors in the intermediate coordinates are combined according to a row of
$A$, and their Euclidean accumulation is controlled by
$\|A\|_{2\to\infty}$.  The resulting space guarantee contains the term
\[
\varepsilon^{-1}
\|A\|_{2\to\infty}\|B\|_{1\to1}.
\]
For $M=Q$, the coordinates of $Qx$ are precisely rank queries, and approximate
rank queries yield approximate quantiles.  Thus a smaller value of
$\gamma_{2,1}(Q)$ would directly improve this factorization-based approach to
turnstile rank and quantile estimation~\cite{BulanekEtAl2026}.  This viewpoint
abstracts the dyadic streaming paradigm used in earlier turnstile quantile
algorithms~\cite{CormodeMuthukrishnan2005,WangEtAl2013,LuoEtAl2016}.

\paragraph{Differential privacy.}
The same factorization has an analogous interpretation in a pure-DP Laplace
matrix mechanism~\cite{LiEtAl2015}.  The mechanism privately measures $Bx$ and releases
$A(Bx+Z)$, where the coordinates of $Z$ are independent Laplace random
variables.  Under unit $\ell_1$ adjacency, the sensitivity of the intermediate
measurement is exactly $\|B\|_{1\to1}$.  Calibrating the Laplace noise to this
sensitivity and reconstructing with $A$ gives maximum per-coordinate mean
squared error
\[
\frac{2}{\varepsilon^2}
\|A\|_{2\to\infty}^2\|B\|_{1\to1}^2.
\]
For $M=Q$, this is the continual-counting workload introduced in the continual
observation model~\cite{DworkEtAl2010,ChanShiSong2011}.  Thus, the same mixed
factorization norm is a key parameter in both settings: its square determines
the error of pure-DP Laplace matrix mechanisms, while the norm itself appears
as a central term in space bounds for factorization-based streaming
algorithms~\cite{BhowmikHasan2026}.  A lower bound on $\gamma_{2,1}(Q)$ therefore gives a
limitation shared by both applications, rather than a lower bound tied to only
one of their surrounding algorithmic models.

\paragraph{Dyadic upper bound.}
There is a simple upper bound from the dyadic decomposition~\cite{CormodeMuthukrishnan2005}.  Place the $n$
coordinates at the leaves of a complete binary tree and use dyadic intervals
as the intermediate dictionary.  Each update belongs to at most one interval
per level, so every column of $B$ has at most $O(\log n)$ nonzero entries.
Every prefix is a disjoint union of at most one dyadic interval per level, so
every row of $A$ has at most $O(\log n)$ nonzero entries.  For the resulting
binary factorization,
\[
\|A\|_{2\to\infty}\le \sqrt{1+\lceil\log_2 n\rceil},
\qquad
\|B\|_{1\to1}\le 1+\lceil\log_2 n\rceil,
\]
and therefore
\[
\gamma_{2,1}(Q)\le (1+\lceil\log_2 n\rceil)^{3/2}.
\]
Fredman gave a signed interval representation with sparsity
$\log_{3+2\sqrt{2}}n$, improving the constant in this construction but not its
$\log^{3/2}n$ asymptotic order~\cite{Fredman82}.  This leads to the central
question:
\begin{center}
\emph{Can the optimal cost over arbitrary real factorizations be
$o(\log^{3/2}n)$?}
\end{center}

Lower bounds for arbitrary real factorizations must accommodate cancellation,
arbitrarily large finite inner dimension, and rescaling of rank-one terms.  Consequently,
neither sparsity nor a bound on the dictionary size is available.  The argument
must instead detect the multiscale discontinuities of the prefix matrix while
coupling the $\ell_2$ geometry of the reconstruction rows with the $\ell_1$
geometry of the dictionary columns.

\subsection{Our Results}
\label{sec:our-results}

Our main result gives a near-matching lower bound for arbitrary real
factorizations of the prefix matrix.

\begin{theorem}[Main theorem]
\label{thm:main_result}
There is an absolute constant $c>0$ such that, for all sufficiently large $n$,
all $m\in\mathbb{N}$, and all real factorizations
$Q=AB$ with $A\in\mathbb{R}^{n\times m}$ and
$B\in\mathbb{R}^{m\times n}$,
\[
\|A\|_{2\to\infty}\|B\|_{1\to1}
\ge
c\frac{\log^{3/2}n}{(\log\log n)^{3/2}}.
\]
Consequently,
\[
\gamma_{2,1}(Q)
=
\Omega\!\left(
\frac{\log^{3/2}n}{(\log\log n)^{3/2}}
\right).
\]
\end{theorem}

The proof of Theorem~\ref{thm:main_result} was obtained using a fully
automated Gemini-based agentic system developed internally at Google.  After we
shared this result with Bulanek et al., they independently obtained the
stronger lower bound
$\Omega(\log^{3/2}n/\log\log n)$ in~\cite{BulanekEtAl2026}.  As recorded in
their AI acknowledgments, their lower-bound proof was generated entirely by an
internal version of the UCLA AI for Math Moonshot harness and was 
rewritten by the authors.  Their result improves our bound by a factor of
$\sqrt{\log\log n}$.  In concurrent work, Bhowmik and Hasan proved the matching
bound $\Theta(\log^{3/2}n)$~\cite{BhowmikHasan2026}, improving our bound by a
factor of $(\log\log n)^{3/2}$.

\subsection{Comparison with Concurrent Work}
\label{sec:concurrent-comparison}

\paragraph{Comparison with Bulanek--Kumar--Meka--Nelson--Sarl\'os.}
In~\cite{BulanekEtAl2026}, Bulanek et al. introduce the factorization norm in
the analysis of $M$-point queries and prove
for the prefix matrix the lower bound
\[
\gamma_{2,1}(Q)
=
\Omega\!\left(\frac{\log^{3/2}n}{\log\log n}\right).
\]
Their proof is organized along root-to-leaf paths of the dyadic interval tree.
After normalizing the column $\ell_1$ norms of the dictionary matrix, they test
the factorization against left-minus-right Haar functionals.  The resulting
constraints pair a vector of dictionary imbalances with a piecewise-affine
tent function.  A martingale argument controls the number of dictionary rows
having large accumulated path variance.  The kink of each tent then supplies a
fresh component orthogonal to all same-or-coarser dyadic directions, and
Bessel's inequality converts the corresponding fresh distances into a lower
bound on a row norm of the reconstruction factor.

Our proof organizes the multiscale information by levels rather than paths.
For a fixed scale $d$, it analyzes all Haar wavelets at that scale
simultaneously through the matrix $X_d=QP_d$.  We partition the dictionary rows
according to their numerical sparsity and use Gram and rank-bounded Schatten
estimates to show that numerically sparse rows cannot supply the required Schatten
$2/3$ mass.  This plays the role that the martingale and fresh-direction
arguments play in~\cite{BulanekEtAl2026}.  Quantitatively, their lower bound is
stronger than ours by a factor of $\sqrt{\log\log n}$.

\paragraph{Comparison with Bhowmik--Hasan.}
In~\cite{BhowmikHasan2026}, Bhowmik and Hasan study the same lower triangular
matrix under both the maximum-row cost
$\|A\|_{2\to\infty}\|B\|_{1\to1}$ and the normalized Frobenius cost.  Their
argument gives the matching order $\Theta(\log^{3/2}n)$ for these costs.  The
central invariant is the finite $p$-nuclear power
\[
\mathsf n_p(M)
=
\inf_{M=\sum_k u_kv_k^T}
\sum_k\bigl(\|u_k\|_2\|v_k\|_1\bigr)^p.
\]
They lower-bound this quantity through aggregate column widths.  For the prefix
matrix,
\[
D_k(Q)=\Theta\!\left(\frac{n^{3/2}}{\sqrt{k}}\right),
\qquad 1\le k\le n/16,
\]
where $D_k(Q)$ is the minimum, over subspaces of dimension at most $k$, of the
sum of the Euclidean distances from the columns of $Q$ to the subspace.  A
rank-one tail estimate and a reverse Hardy accumulation transfer the width
profile to $\mathsf n_p(Q)$.  At the critical exponent $p=2/3$, the resulting
summand indexed by $k$ is of order $n/k$, producing
$\mathsf n_{2/3}(Q)=\Theta(n\log n)$.  H\"older's inequality then transfers
this bound to the two factorization costs, and a Fenwick interval factorization
supplies the matching upper bounds.

There is a direct algebraic connection between this invariant and the initial
reduction in our proof.  If $Q=AB$, with $a_k$ the $k$-th column of
$A$ and $b_k^T$ the $k$-th row of $B$, then
\[
\sum_k
\bigl(\|a_k\|_2\|b_k\|_1\bigr)^{2/3}
=
\sum_k N_k^{1/3}L_k^{2/3}.
\]
Thus the fractional rank-one cost used below is precisely the $2/3$-power
cost of the rank-one representation induced by the factorization.  The
aggregate-width method gives a lower bound of order $n\log n$ for this sum,
whereas our scale-wise numerical-sparsity method gives
$n\log n/\log\log n$.  H\"older's inequality converts a lower bound $S$ for
the sum into the factorization lower bound $(S/n)^{3/2}$.  Consequently, the
factor $\log\log n$ lost in our sum becomes a factor
$(\log\log n)^{3/2}$ between the final lower bounds.

\section{Proof of the Main Theorem}
\label{sec:proof}

We give the proof of Theorem~\ref{thm:main_result}.  The argument combines
scale-wise Haar projections, a numerical-sparsity dichotomy for the rows of
the dictionary factor, and an aggregation over dyadic scales.

\subsection{Fractional H\"older Reduction}
\label{sec:fractional-reduction}
The primary obstruction to proving a sharp lower bound on the target factorization norm is that the $\ell_2$ penalty only grows as $O(\log n)$, while the $\ell_1$ sparsity constraint interacts globally with the structure of $B$. In this subsection, we provide a rigorous algebraic reduction that decouples the factorization energy from the global max-norms, isolating the required spatial-spectral trade-off into a single additive cost function over the intermediate vectors of the factorization.

Let $Q = AB$ be any exact factorization of the $n \times n$ lower triangular all-ones matrix, where $A \in \mathbb{R}^{n \times m}$ and $B \in \mathbb{R}^{m \times n}$. We are interested in the mixed factorization norm $M_A M_B$, where 
\begin{align*}
M_A &= \|A\|_{2\to\infty} = \max_{1 \le i \le n} \|A_{i,:}\|_2, \\
M_B &= \|B\|_{1\to 1} = \max_{1 \le j \le n} \|B_{:,j}\|_1.
\end{align*}

\begin{definition}[Column Energy and Row Capacity]
\label{def:volume_metrics}
For each $k \in \{1, \dots, m\}$, let $N_k$ denote the total squared $\ell_2$ energy of the $k$-th column of $A$, and let $L_k$ denote the total $\ell_1$ capacity of the $k$-th row of $B$:
\begin{align*}
N_k &= \|A_{:,k}\|_2^2 = \sum_{i=1}^n A_{i,k}^2, \\
L_k &= \|B_{k,:}\|_1 = \sum_{j=1}^n |B_{k,j}|.
\end{align*}
\end{definition}

\begin{lemma}[Global Factor-Sum Bounds]
\label{lem:volume_constraints}
For any factorization $Q=AB$, the sequences $\{N_k\}_{k=1}^m$ and $\{L_k\}_{k=1}^m$ satisfy the following global factor-sum bounds:
\begin{align}
\sum_{k=1}^m N_k &\le n M_A^2, \label{eq:volume_N} \\
\sum_{k=1}^m L_k &\le n M_B. \label{eq:volume_L}
\end{align}
\end{lemma}
\begin{proof}
For the squared column norms $N_k$, we exchange the order of summation to express the total sum in terms of the row norms of $A$:
\begin{equation*}
\sum_{k=1}^m N_k = \sum_{k=1}^m \sum_{i=1}^n A_{i,k}^2 = \sum_{i=1}^n \sum_{k=1}^m A_{i,k}^2 = \sum_{i=1}^n \|A_{i,:}\|_2^2.
\end{equation*}
Since each of the $n$ rows of $A$ has a squared $\ell_2$ norm bounded by $M_A^2$, the sum over all $n$ rows is at most $n M_A^2$.

Similarly, for the row $\ell_1$ norms $L_k$, we exchange the order of summation to express the total sum in terms of the column $\ell_1$ norms of $B$:
\begin{equation*}
\sum_{k=1}^m L_k = \sum_{k=1}^m \sum_{j=1}^n |B_{k,j}| = \sum_{j=1}^n \sum_{k=1}^m |B_{k,j}| = \sum_{j=1}^n \|B_{:,j}\|_1.
\end{equation*}
Since each of the $n$ columns of $B$ has an $\ell_1$ norm bounded by $M_B$, the sum over all $n$ columns is at most $n M_B$.
\end{proof}

\begin{lemma}[Nontriviality of the Factorization Cost]
\label{lem:nontrivial_cost}
For every exact factorization $Q=AB$ of the lower triangular all-ones matrix,
\[
M_A M_B \ge 1.
\]
Consequently, $\Delta=M_A^2M_B^2\ge 1$.
\end{lemma}

\begin{proof}
For any entry with $i\ge j$, we have $Q_{ij}=1$. Hence
\[
1 = |Q_{ij}| = |\langle A_{i,:},B_{:,j}\rangle|
\le \|A_{i,:}\|_2 \|B_{:,j}\|_2
\le \|A_{i,:}\|_2 \|B_{:,j}\|_1
\le M_A M_B.
\]
Squaring gives $\Delta=M_A^2M_B^2\ge 1$.
\end{proof}

\begin{theorem}[Fractional H\"older Reduction]
\label{thm:fractional_holder}
For every factorization $Q=AB$, the fractional rank-one cost of the induced representation satisfies
\begin{equation*}
\sum_{k=1}^m N_k^{1/3} L_k^{2/3} \le n(M_A M_B)^{2/3}.
\end{equation*}
Consequently, if the Fractional Rank-One Bound
\begin{equation}
\label{eq:fractional_cost_target}
\sum_{k=1}^m N_k^{1/3} L_k^{2/3} \ge \Omega\left(n \frac{\log n}{\log \log n}\right)
\end{equation}
holds, then the factorization norm is bounded by 
\begin{equation*}
M_A M_B \ge \Omega\left(\frac{\log^{1.5} n}{(\log \log n)^{1.5}}\right).
\end{equation*}
\end{theorem}
\begin{proof}
We apply H\"older's inequality to the sequence of products $N_k^{1/3} L_k^{2/3}$. Choosing conjugate exponents $p=3$ and $q=3/2$, which satisfy $1/p + 1/q = 1/3 + 2/3 = 1$, we obtain:
\begin{equation*}
\sum_{k=1}^m N_k^{1/3} L_k^{2/3} \le \left( \sum_{k=1}^m \left(N_k^{1/3}\right)^3 \right)^{1/3} \left( \sum_{k=1}^m \left(L_k^{2/3}\right)^{3/2} \right)^{2/3}.
\end{equation*}
Simplifying the exponents yields:
\begin{equation*}
\sum_{k=1}^m N_k^{1/3} L_k^{2/3} \le \left( \sum_{k=1}^m N_k \right)^{1/3} \left( \sum_{k=1}^m L_k \right)^{2/3}.
\end{equation*}
Substituting the global factor-sum bounds from Lemma~\ref{lem:volume_constraints} gives:
\begin{equation*}
\sum_{k=1}^m N_k^{1/3} L_k^{2/3} \le \left( n M_A^2 \right)^{1/3} \left( n M_B \right)^{2/3} = n^{1/3 + 2/3} \left( M_A^{2/3} \right) \left( M_B^{2/3} \right) = n (M_A M_B)^{2/3}.
\end{equation*}
Assuming the target fractional rank-one bound holds (Equation~\ref{eq:fractional_cost_target}), we substitute the lower bound:
\begin{equation*}
\Omega\left(n \frac{\log n}{\log \log n}\right) \le n (M_A M_B)^{2/3}.
\end{equation*}
Dividing both sides by $n$, we obtain:
\begin{equation*}
(M_A M_B)^{2/3} \ge \Omega\left(\frac{\log n}{\log \log n}\right).
\end{equation*}
Raising both sides to the power of $3/2$ yields:
\begin{equation*}
M_A M_B \ge \Omega\left(\frac{\log^{1.5} n}{(\log \log n)^{1.5}}\right).
\end{equation*}
This establishes the unconditional reduction.
\end{proof}

\begin{remark}
This reduction translates the coupled min-max optimization of $\|A\|_{2\to\infty} \|B\|_{1\to 1}$ into an uncoupled, additive penalty function over the intermediate dimensions. The remainder of this paper is dedicated to proving that the multiscale spectral synthesis of $Q$ demands an aggregate fractional capacity of $\Omega(n \log n / \log \log n)$, which completes the unconditional proof of the barrier.
\end{remark}

\subsection{Direct Right-Sided Haar Projection and Schatten 2/3 Mass}
\label{sec:haar-projection}
To isolate the contribution of $Q$ at scale $d$, we project on the right onto
the Haar subspace at that scale and set $X_d=QP_d$.  For fixed $d$, the
vectors $Qh_{d,k}$ have disjoint supports, which makes the nonzero singular
values of $X_d$ explicit.

\begin{definition}[Multiscale Haar Basis and Projection]
\label{def:haar_basis}
Restricting to the leading principal submatrix of order
$2^{\lfloor\log_2 n\rfloor}$ cannot increase
$\|A\|_{2\to\infty}$ or $\|B\|_{1\to1}$, so we may assume that $n$ is a
power of two. For each scale $d \in \{1, 2, \dots, \log_2 n\}$, let $w_d = 2^d$ denote the spatial support size of the wavelets at scale $d$. There are $n/w_d$ mutually orthogonal wavelets at this scale, denoted $\{h_{d,k}\}_{k=1}^{n/w_d}$, where each $h_{d,k} \in \mathbb{R}^n$ is supported on the interval $I_{d,k} = \{(k-1)w_d + 1, \dots, k w_d\}$. The vector entries are defined as:
\begin{equation}
h_{d,k}(i) = 
\begin{cases} 
\frac{1}{\sqrt{w_d}} & \text{if } i \in \left[(k-1)w_d + 1, (k-1)w_d + \frac{w_d}{2}\right], \\
-\frac{1}{\sqrt{w_d}} & \text{if } i \in \left[(k-1)w_d + \frac{w_d}{2} + 1, k w_d\right], \\
0 & \text{otherwise.}
\end{cases}
\end{equation}
Let $P_d = \sum_{k=1}^{n/w_d} h_{d,k} h_{d,k}^T$ be the orthogonal projection onto the subspace spanned by the Haar wavelets at scale $d$.
\end{definition}

We now evaluate the direct right-sided projection $X_d = Q P_d$. Because $Q$ operates as a prefix-sum matrix, its action on the Haar basis is structured.

\begin{lemma}[Action of $Q$ on the Haar Basis]
\label{lem:Q_action_haar}
For any scale $d \ge 1$ and index $k \in \{1, \dots, n/w_d\}$, the projected vector $g_{d,k} = Q h_{d,k}$ is supported on the interval $I_{d,k} \setminus \{k w_d\}$. Furthermore, the set of vectors $\{g_{d,k}\}_{k=1}^{n/w_d}$ is mutually orthogonal, and the squared $\ell_2$ norm of each vector is exactly
\begin{equation}
\|g_{d,k}\|_2^2 = \frac{w_d^2+2}{12}.
\end{equation}
\end{lemma}
\begin{proof}
Because $Q$ is the lower triangular all-ones matrix, the vector $g_{d,k} = Q h_{d,k}$ contains the prefix sums of $h_{d,k}$. That is, $g_{d,k}(i) = \sum_{j=1}^i h_{d,k}(j)$. 
For $i < (k-1)w_d + 1$, the sum is trivially $0$.
Within the first half of the interval, $i \in \left[(k-1)w_d + 1, (k-1)w_d + w_d/2\right]$, the values accumulate linearly: letting $j = i - (k-1)w_d$, we have $g_{d,k}(i) = \frac{j}{\sqrt{w_d}}$.
In the second half, $i \in \left[(k-1)w_d + w_d/2 + 1, k w_d\right]$, the values decrease linearly from the peak: letting $j = i - ((k-1)w_d + w_d/2)$, we have $g_{d,k}(i) = \frac{w_d/2 - j}{\sqrt{w_d}}$.
For $i \ge k w_d$, the sum over the entire Haar wavelet is exactly $0$ because the wavelet has mean zero. Therefore, $g_{d,k}$ is supported on $I_{d,k} \setminus \{k w_d\}$. 
Since the support intervals $I_{d,k}$ are mutually disjoint for different $k$ at a fixed scale $d$, the vectors $g_{d,k}$ are mutually orthogonal.
To compute the squared $\ell_2$ norm, we sum the squared components. Letting $m = w_d/2$:
\begin{align*}
\|g_{d,k}\|_2^2 &= \sum_{j=1}^m \left(\frac{j}{\sqrt{w_d}}\right)^2 + \sum_{j=1}^{m-1} \left(\frac{m - j}{\sqrt{w_d}}\right)^2 \\
&= \frac{1}{w_d} \left( \sum_{j=1}^m j^2 + \sum_{j=1}^{m-1} j^2 \right) \\
&= \frac{1}{2m} \left( \frac{m(m+1)(2m+1)}{6} + \frac{(m-1)m(2m-1)}{6} \right).
\end{align*}
Expanding the terms inside the parenthesis gives $m(2m^2 + 3m + 1 + 2m^2 - 3m + 1) = m(4m^2 + 2)$.
Thus, 
\begin{equation}
\|g_{d,k}\|_2^2 = \frac{1}{2m} \frac{m(4m^2+2)}{6} = \frac{4m^2+2}{12} = \frac{w_d^2+2}{12}.
\end{equation}
\end{proof}

To quantify the energy of $X_d$, we use the Schatten $p$-norms. Because we are targeting the $(M_A M_B)^{2/3}$ factorization cost (via the Fractional H\"older Bound in Section~\ref{sec:fractional-reduction}), we must apply the subadditivity of the Schatten $2/3$ quasi-norm.

\begin{theorem}[$p$-Triangle Inequality for Schatten $2/3$-Norm]
\label{thm:p_triangle}
For any matrices $Y$ and $Z$ of the same dimensions, the Schatten $2/3$
quasi-norm, defined by $\|M\|_{S_{2/3}} = (\sum \sigma_j(M)^{2/3})^{3/2}$,
satisfies the $p$-triangle inequality for $p=2/3$~\cite{Thompson1976}:
\begin{equation}
\|Y + Z\|_{S_{2/3}}^{2/3} \le \|Y\|_{S_{2/3}}^{2/3} + \|Z\|_{S_{2/3}}^{2/3}.
\end{equation}
\end{theorem}

\begin{theorem}[Exact Schatten 2/3 Mass of the Direct Projection]
\label{thm:Xd_schatten}
For any scale $d \in \{1, \dots, \log_2 n\}$, the direct right-sided projection $X_d = Q P_d$ has exactly $n/w_d$ non-zero singular values, all equal to $\sigma_d = \sqrt{\frac{w_d^2+2}{12}} = \Theta(w_d)$. Consequently, its Schatten $2/3$ mass is $\Theta(n w_d^{-1/3})$. Specifically, there exists an absolute constant $c_0 = (1/12)^{1/3} > 0$ such that:
\begin{equation}
\|X_d\|_{S_{2/3}}^{2/3} \ge c_0 n w_d^{-1/3}.
\end{equation}
\end{theorem}
\begin{proof}
We can explicitly expand the operator $X_d$ using the orthonormal basis for the range of $P_d$:
\begin{equation}
X_d = Q \left( \sum_{k=1}^{n/w_d} h_{d,k} h_{d,k}^T \right) = \sum_{k=1}^{n/w_d} (Q h_{d,k}) h_{d,k}^T = \sum_{k=1}^{n/w_d} g_{d,k} h_{d,k}^T.
\end{equation}
By defining the normalized vectors $u_{d,k} = g_{d,k} / \|g_{d,k}\|_2$, we can rewrite $X_d$ as:
\begin{equation}
X_d = \sum_{k=1}^{n/w_d} \|g_{d,k}\|_2 \, u_{d,k} h_{d,k}^T.
\end{equation}
By Lemma \ref{lem:Q_action_haar}, the set $\{u_{d,k}\}_{k=1}^{n/w_d}$ is orthonormal. By definition, $\{h_{d,k}\}_{k=1}^{n/w_d}$ is also an orthonormal set. Thus, this sum constitutes the compact Singular Value Decomposition of $X_d$. 
The non-zero singular values are exactly $\sigma_d = \|g_{d,k}\|_2 = \sqrt{\frac{w_d^2+2}{12}}$, and there are exactly $n/w_d$ such values.
Summing the $2/3$ powers of these singular values gives the exact Schatten $2/3$ mass:
\begin{equation}
\|X_d\|_{S_{2/3}}^{2/3} = \frac{n}{w_d} \left( \frac{w_d^2+2}{12} \right)^{1/3} = n \left( \frac{1}{12} \right)^{1/3} w_d^{-1/3} \left(1 + \frac{2}{w_d^2}\right)^{1/3}.
\end{equation}
Since $w_d \ge 2$, the term $(1 + 2/w_d^2)^{1/3}$ is greater than $1$. We can therefore bound the mass from below by setting $c_0 = (1/12)^{1/3}$:
\begin{equation}
\|X_d\|_{S_{2/3}}^{2/3} \ge c_0 n w_d^{-1/3}.
\end{equation}
Because the upper bound is also bounded by a constant factor for all $w_d \ge 2$, the mass is $\Theta(n w_d^{-1/3})$.
\end{proof}

\subsection{Numerical Sparsity and Row-Submatrix Spectral Bounds}
\label{sec:numerical-sparsity}
In order to connect the local $\ell_2$ capacity of the factorization factors with their global $\ell_1$ footprint, we use the numerical sparsity of the rows of $B$. This quantity measures effective support size and will allow us to bound the spectral norm of any sub-frame using matrix norm interpolation.

\begin{definition}[Numerical Sparsity]
\label{def:numerical_sparsity}
For any row vector $b_k = (B_{k,:})^T \in \mathbb{R}^n$ of the factorization factor $B$, its \emph{numerical sparsity} $S_k$ is defined as the ratio of its squared $\ell_1$ norm to its squared $\ell_2$ norm:
\begin{equation}
S_k = \frac{\|b_k\|_1^2}{\|b_k\|_2^2} = \frac{L_k^2}{\|b_k\|_2^2},
\end{equation}
where $L_k = \|b_k\|_1$. If $b_k = 0$, we define $S_k = 1$.

By the standard relations of $\ell_p$ norms in $\mathbb{R}^n$, we always have $1 \le S_k \le n$. The parameter $S_k$ is an effective support-size measure: $S_k = \Theta(1)$ for numerically sparse vectors (e.g., standard basis vectors), while $S_k = \Theta(n)$ for flat vectors.
\end{definition}

\begin{lemma}[Row Capacity Bound via Numerical Sparsity]
\label{lem:row_capacity_sparsity}
For any row $k$, its $\ell_1$ capacity $L_k$ is bounded by its numerical sparsity $S_k$ and the global column max-norm $M_B = \|B\|_{1 \to 1}$:
\begin{equation}
L_k \le M_B S_k.
\end{equation}
\end{lemma}
\begin{proof}
Let $B_{k,j}$ denote the entries of $B$. The definition of $M_B = \max_j \sum_i |B_{i,j}|$ implies that every individual entry is bounded in absolute value by $M_B$, i.e., $|B_{k,j}| \le M_B$ for all $k, j$.

We bound the squared $\ell_2$ norm of the $k$-th row by factoring out the maximum absolute entry:
\begin{equation}
\|b_k\|_2^2 = \sum_{j=1}^n B_{k,j}^2 \le \sum_{j=1}^n |B_{k,j}| M_B = M_B \|b_k\|_1 = M_B L_k.
\end{equation}

Substituting this upper bound into the definition of numerical sparsity yields:
\begin{equation}
S_k = \frac{L_k^2}{\|b_k\|_2^2} \ge \frac{L_k^2}{M_B L_k} = \frac{L_k}{M_B}.
\end{equation}
Rearranging the inequality gives $L_k \le M_B S_k$. (If $L_k = 0$, the inequality holds trivially).
\end{proof}

\begin{theorem}[Spectral Bound for a Row Submatrix]
\label{thm:gram_spectral_bound}
Let $B_{sub}$ be any submatrix of $B$ formed by selecting an arbitrary subset $\mathcal{I} \subseteq \{1, \dots, m\}$ of its rows. Then the maximum eigenvalue of its Gram matrix $P_{sub} = B_{sub}^T B_{sub}$ is bounded by:
\begin{equation}
\lambda_{\max}(B_{sub}^T B_{sub}) \le M_B^2 \max_{k \in \mathcal{I}} S_k.
\end{equation}
\end{theorem}
\begin{proof}

The maximum eigenvalue of the positive semi-definite matrix $B_{sub}^T B_{sub}$ is exactly the squared spectral operator norm of $B_{sub}$:
\begin{equation}
\lambda_{\max}(B_{sub}^T B_{sub}) = \|B_{sub}\|_2^2.
\end{equation}

By Schur's test (or the Riesz-Thorin interpolation theorem), the $\ell_2 \to \ell_2$ operator norm of any matrix is bounded by the product of its $\ell_1 \to \ell_1$ norm (maximum absolute column sum) and its $\ell_\infty \to \ell_\infty$ norm (maximum absolute row sum):
\begin{equation}
\|B_{sub}\|_2^2 \le \|B_{sub}\|_{1 \to 1} \|B_{sub}\|_{\infty \to \infty}.
\end{equation}

We can verify this. For any $x \in \mathbb{R}^n$, by the Cauchy-Schwarz inequality on the inner sums:
\begin{align}
\|B_{sub} x\|_2^2 &= \sum_{k \in \mathcal{I}} \left( \sum_{j=1}^n B_{k,j} x_j \right)^2 \le \sum_{k \in \mathcal{I}} \left( \sum_{j=1}^n |B_{k,j}| |x_j| \right)^2 \nonumber \\
&\le \sum_{k \in \mathcal{I}} \left( \sum_{j=1}^n |B_{k,j}| \right) \left( \sum_{j=1}^n |B_{k,j}| x_j^2 \right) \nonumber \\
&\le \|B_{sub}\|_{\infty \to \infty} \sum_{k \in \mathcal{I}} \sum_{j=1}^n |B_{k,j}| x_j^2 \nonumber \\
&= \|B_{sub}\|_{\infty \to \infty} \sum_{j=1}^n x_j^2 \left( \sum_{k \in \mathcal{I}} |B_{k,j}| \right) \nonumber \\
&\le \|B_{sub}\|_{\infty \to \infty} \|B_{sub}\|_{1 \to 1} \|x\|_2^2.
\end{align}

Because $B_{sub}$ is a submatrix of $B$ containing a subset of its rows, its maximum absolute column sum cannot exceed that of the full matrix $B$. Thus,
\begin{equation}
\|B_{sub}\|_{1 \to 1} \le \|B\|_{1 \to 1} = M_B.
\end{equation}

The maximum absolute row sum of $B_{sub}$ is exactly the maximum $\ell_1$ capacity among the selected rows:
\begin{equation}
\|B_{sub}\|_{\infty \to \infty} = \max_{k \in \mathcal{I}} \|b_k\|_1 = \max_{k \in \mathcal{I}} L_k.
\end{equation}

Applying Lemma \ref{lem:row_capacity_sparsity}, we bound the capacity of any row $k \in \mathcal{I}$ by $L_k \le M_B S_k$. Therefore:
\begin{equation}
\|B_{sub}\|_{\infty \to \infty} \le M_B \max_{k \in \mathcal{I}} S_k.
\end{equation}

Substituting both of these structural bounds back into Schur's test yields:
\begin{equation}
\|B_{sub}\|_2^2 \le M_B \left( M_B \max_{k \in \mathcal{I}} S_k \right) = M_B^2 \max_{k \in \mathcal{I}} S_k.
\end{equation}
This completes the proof.
\end{proof}

\begin{remark}
Theorem \ref{thm:gram_spectral_bound} bridges the gap between the $\ell_1$ geometry of $B$ and the spectral $\ell_2$ capacity required to factorize $Q$. If an adversarial factorization attempts to deploy a sub-frame of numerically sparse vectors (where $S_k \le O(1)$), the row-submatrix spectral bound constrains the squared spectral norm of this sub-frame to $O(M_B^2)$. This rigid constraint limits how much Frobenius mass numerically sparse vectors can extract from any individual scale, forcing the factorization to rely on diffuse capacity to satisfy the matrix step singularities.
\end{remark}

\subsection{Numerically Sparse-Part Frobenius Bound}
\label{sec:trace-bottleneck}
To dissect the cost of synthesizing the lower triangular matrix $Q$, we analyze the factorization at a fixed spatial scale $d \in \{1, \dots, \log_2 n\}$ with spatial width $w_d = 2^d$. We partition the intermediate vectors into two classes based on their numerical sparsity $S_k$: those below the scale-dependent threshold (numerically sparse) and those above it (diffuse).

\begin{definition}[Sparse--Diffuse Partition]
\label{def:sparse_diffuse_partition}
Let $Q = AB$ be a matrix factorization with maximal norms $M_A$ and $M_B$. We define the factorization capacity parameter $\Delta = M_A^2 M_B^2$. Since $Q$ is non-zero, any exact factorization requires $M_A > 0$ and $M_B > 0$, ensuring $\Delta > 0$.
For a given scale $d$ and any tuning parameter $\epsilon > 0$, we partition the intermediate dimensions $k \in \{1, \dots, m\}$ into two disjoint sets:
\begin{align}
\mathcal{I}_{\mathrm{sparse}}(d) &= \left\{ k \in \{1, \dots, m\} \;\middle|\; S_k \le \frac{\epsilon w_d}{\Delta} \right\}, \\
\mathcal{I}_{\mathrm{diffuse}}(d) &= \left\{ k \in \{1, \dots, m\} \;\middle|\; S_k > \frac{\epsilon w_d}{\Delta} \right\},
\end{align}
where $S_k = \|B_{k,:}\|_1^2 / \|B_{k,:}\|_2^2$ is the numerical sparsity defined in Definition~\ref{def:numerical_sparsity}.
\end{definition}

Corresponding to this partition, we split the factors $A$ and $B$. Let $A_{\mathrm{sparse}}$ and $B_{\mathrm{sparse}}$ be the submatrices consisting of the columns of $A$ and rows of $B$, respectively, indexed by $\mathcal{I}_{\mathrm{sparse}}(d)$. Let $A_{\mathrm{diffuse}}$ and $B_{\mathrm{diffuse}}$ be the complementary submatrices indexed by $\mathcal{I}_{\mathrm{diffuse}}(d)$. 

The direct right-sided Haar projection $X_d = Q P_d$ at scale $d$ decomposes into the sum of the sparse and diffuse projections:
\begin{equation}
X_d = A B P_d = A_{\mathrm{sparse}} B_{\mathrm{sparse}} P_d + A_{\mathrm{diffuse}} B_{\mathrm{diffuse}} P_d = X_{\mathrm{sparse}} + X_{\mathrm{diffuse}}.
\end{equation}

Our goal is to bound the Frobenius mass of the sparse projection $X_{\mathrm{sparse}}$. The squared Frobenius norm, $\|X_{\mathrm{sparse}}\|_F^2 = \text{tr}(X_{\mathrm{sparse}}^T X_{\mathrm{sparse}})$, captures the total $\ell_2$ energy that the rows below the numerical-sparsity threshold can contribute to the Haar scale $d$. 

\begin{lemma}[Submultiplicativity of Frobenius and Spectral Norms]
\label{lem:frob_spectral_submult}
For any real matrices $U$ and $V$ of compatible dimensions, the Frobenius norm of their product is bounded by the product of the Frobenius norm of $U$ and the spectral norm of $V$:
\begin{equation}
\|U V\|_F \le \|U\|_F \|V\|_2.
\end{equation}
\end{lemma}
\begin{proof}
By the definition of the Frobenius norm via the trace operator and the cyclic property of the trace, we have:
\begin{equation*}
\|UV\|_F^2 = \text{tr}((UV)^T (UV)) = \text{tr}(V^T U^T U V) = \text{tr}(U^T U V V^T).
\end{equation*}
Since $U^T U$ and $V V^T$ are positive semi-definite matrices, we can apply the trace inequality $\text{tr}(Y Z) \le \lambda_{\max}(Z) \text{tr}(Y)$ for positive semi-definite matrices $Y$ and $Z$. Setting $Y = U^T U$ and $Z = V V^T$, we obtain:
\begin{equation*}
\text{tr}(U^T U V V^T) \le \lambda_{\max}(V V^T) \text{tr}(U^T U) = \|V\|_2^2 \|U\|_F^2.
\end{equation*}
Taking the square root yields $\|UV\|_F \le \|U\|_F \|V\|_2$.
\end{proof}

\begin{theorem}[Numerically Sparse-Part Frobenius Bound]
\label{thm:sparse_frob_bound}
For any scale $d$ and tuning parameter $\epsilon > 0$, the squared Frobenius mass of the sparse projection is bounded by:
\begin{equation}
\|X_{\mathrm{sparse}}\|_F^2 \le \epsilon n w_d.
\end{equation}
\end{theorem}
\begin{proof}
If $\mathcal{I}_{\mathrm{sparse}}(d)$ is empty, $X_{\mathrm{sparse}} = 0$ and the bound holds trivially. Assume $\mathcal{I}_{\mathrm{sparse}}(d)$ is non-empty.

We apply Lemma~\ref{lem:frob_spectral_submult} to the product $X_{\mathrm{sparse}} = A_{\mathrm{sparse}} (B_{\mathrm{sparse}} P_d)$:
\begin{equation}
\label{eq:Xsparse_frob_bound}
\|X_{\mathrm{sparse}}\|_F^2 \le \|A_{\mathrm{sparse}}\|_F^2 \|B_{\mathrm{sparse}} P_d\|_2^2.
\end{equation}
Because $P_d$ is an orthogonal projection operator, its spectral norm is $\|P_d\|_2 = 1$. By the submultiplicativity of the spectral norm, we bound the operator norm of the right factor:
\begin{equation}
\|B_{\mathrm{sparse}} P_d\|_2^2 \le \|B_{\mathrm{sparse}}\|_2^2 \|P_d\|_2^2 = \|B_{\mathrm{sparse}}\|_2^2.
\end{equation}
By the Spectral Bound for a Row Submatrix (Theorem~\ref{thm:gram_spectral_bound}), the maximum eigenvalue of the Gram matrix $B_{\mathrm{sparse}}^T B_{\mathrm{sparse}}$ (which equals $\|B_{\mathrm{sparse}}\|_2^2$) is constrained by the maximum numerical sparsity among its rows:
\begin{equation}
\|B_{\mathrm{sparse}}\|_2^2 \le M_B^2 \max_{k \in \mathcal{I}_{\mathrm{sparse}}(d)} S_k.
\end{equation}
By the definition of the sparse--diffuse partition (Definition~\ref{def:sparse_diffuse_partition}), every index $k \in \mathcal{I}_{\mathrm{sparse}}(d)$ satisfies $S_k \le \frac{\epsilon w_d}{\Delta}$. Therefore:
\begin{equation}
\label{eq:Bsparse_spectral_bound}
\|B_{\mathrm{sparse}} P_d\|_2^2 \le M_B^2 \left( \frac{\epsilon w_d}{\Delta} \right).
\end{equation}

Next, we bound the Frobenius norm of $A_{\mathrm{sparse}}$. The squared Frobenius norm is exactly the sum of the squared $\ell_2$ norms of the columns of $A_{\mathrm{sparse}}$. Since $A_{\mathrm{sparse}}$ consists of a subset of the columns of $A$, this sum cannot exceed the total squared $\ell_2$ energy of all columns of $A$. Recalling the definition $N_k = \|A_{:,k}\|_2^2$ and using the Global Factor-Sum Bounds (Lemma~\ref{lem:volume_constraints}), we obtain:
\begin{equation}
\label{eq:Asparse_frob_bound}
\|A_{\mathrm{sparse}}\|_F^2 = \sum_{k \in \mathcal{I}_{\mathrm{sparse}}(d)} N_k \le \sum_{k=1}^m N_k \le n M_A^2.
\end{equation}

Substituting the spectral bound \eqref{eq:Bsparse_spectral_bound} and the Frobenius bound \eqref{eq:Asparse_frob_bound} into our initial inequality \eqref{eq:Xsparse_frob_bound} yields:
\begin{equation}
\|X_{\mathrm{sparse}}\|_F^2 \le (n M_A^2) \left( M_B^2 \frac{\epsilon w_d}{\Delta} \right) = n M_A^2 M_B^2 \frac{\epsilon w_d}{\Delta}.
\end{equation}
Finally, substituting the factorization capacity parameter $\Delta = M_A^2 M_B^2$ cancels the maximal norms, giving exactly:
\begin{equation}
\|X_{\mathrm{sparse}}\|_F^2 \le \epsilon n w_d.
\end{equation}
This completes the proof.
\end{proof}

\begin{corollary}[Rank of the Sparse Projection]
\label{cor:rank_sparse}
The rank of the sparse projection $X_{\mathrm{sparse}}$ is bounded by the dimension of the Haar wavelet subspace at scale $d$:
\begin{equation}
\text{rank}(X_{\mathrm{sparse}}) \le \frac{n}{w_d}.
\end{equation}
\end{corollary}
\begin{proof}
Because $X_{\mathrm{sparse}} = A_{\mathrm{sparse}} B_{\mathrm{sparse}} P_d$, its row space is contained within the image of $P_d$. By Definition~\ref{def:haar_basis}, $P_d$ is the orthogonal projection onto the $n/w_d$ mutually orthogonal Haar wavelets at scale $d$. Therefore, the rank of $P_d$ is exactly $n/w_d$, which guarantees that $\text{rank}(X_{\mathrm{sparse}}) \le \text{rank}(P_d) = n/w_d$.
\end{proof}

\begin{remark}
Theorem~\ref{thm:sparse_frob_bound} and Corollary~\ref{cor:rank_sparse} give the numerically sparse-part Frobenius and rank bounds: despite having access to the full $\ell_2$ volume of the factors, numerically sparse vectors are spectrally capped by their small numerical sparsity. The Frobenius mass they can project onto scale $d$ is bounded by $\epsilon n w_d$. Because the direct projection $X_d$ has a macroscopic Schatten 2/3 mass of $\Theta(n w_d^{-1/3})$ (as established in Section~\ref{sec:haar-projection}), the factorization must deploy diffuse capacity to synthesize the matrix.
\end{remark}

\subsection{Rank--Frobenius to Schatten Transfer and Diffuse-Part Mass}
\label{sec:schatten-transfer}
In Section~\ref{sec:trace-bottleneck}, we established that vectors below the numerical-sparsity threshold are spectrally capped: their small numerical sparsity limits the Frobenius mass they can project onto any spatial scale. We now formalize a mechanism to translate this Frobenius bound into a Schatten $2/3$ quasi-norm bound. Because numerically sparse vectors operate within a restricted low-rank subspace defined by the multiscale Haar expansion, their ability to inflate the Schatten $2/3$ mass using many small singular values is blocked. This forces the diffuse components of the factorization to carry a fixed positive fraction of the required Schatten $2/3$ mass.

\begin{lemma}[Rank--Frobenius to Schatten Transfer]
\label{lem:rank_schatten_transfer}
For any real matrix $M$ with rank at most $r$ and squared Frobenius norm $\|M\|_F^2 \le V$, its Schatten $2/3$ quasi-norm evaluated to the $2/3$ power is bounded by:
\begin{equation}
\|M\|_{S_{2/3}}^{2/3} \le V^{1/3} r^{2/3}.
\end{equation}
\end{lemma}
\begin{proof}
Let $\sigma_1, \dots, \sigma_k$ be the non-zero singular values of $M$, where $k \le r$. By definition, the squared Frobenius norm is the sum of squared singular values, $\sum_{i=1}^k \sigma_i^2 = \|M\|_F^2 \le V$, and the targeted Schatten mass is $\sum_{i=1}^k \sigma_i^{2/3}$.

We apply H\"older's inequality to the sum $\sum_{i=1}^k \sigma_i^{2/3} \cdot 1$ with conjugate exponents $p=3$ and $q=3/2$ (which satisfy $1/p + 1/q = 1/3 + 2/3 = 1$):
\begin{equation*}
\sum_{i=1}^k \sigma_i^{2/3} \cdot 1 \le \left( \sum_{i=1}^k \left(\sigma_i^{2/3}\right)^3 \right)^{1/3} \left( \sum_{i=1}^k 1^{3/2} \right)^{2/3}.
\end{equation*}
Simplifying the terms inside the sums yields:
\begin{equation*}
\sum_{i=1}^k \sigma_i^{2/3} \le \left( \sum_{i=1}^k \sigma_i^2 \right)^{1/3} k^{2/3}.
\end{equation*}
Substituting the Frobenius bound $\sum_{i=1}^k \sigma_i^2 \le V$ and bounding the rank $k \le r$, we obtain:
\begin{equation*}
\|M\|_{S_{2/3}}^{2/3} \le V^{1/3} r^{2/3}.
\end{equation*}
This completes the proof.
\end{proof}

\begin{theorem}[Sparse-Part Schatten Mass Bound]
\label{thm:sparse_schatten_bound}
For any scale $d \in \{1, \dots, \log_2 n\}$ with wavelet width $w_d$ and any tuning parameter $\epsilon > 0$, the Schatten $2/3$ mass of the sparse projection $X_{\mathrm{sparse}}$ is bounded by:
\begin{equation}
\|X_{\mathrm{sparse}}\|_{S_{2/3}}^{2/3} \le \epsilon^{1/3} n w_d^{-1/3}.
\end{equation}
\end{theorem}
\begin{proof}
From Section~\ref{sec:trace-bottleneck}, the sparse projection $X_{\mathrm{sparse}} = A_{\mathrm{sparse}} B_{\mathrm{sparse}} P_d$ satisfies two structural bounds:
\begin{enumerate}
    \item By Corollary~\ref{cor:rank_sparse}, its rank is bounded by the dimension of the right-sided Haar wavelet subspace at scale $d$: $\text{rank}(X_{\mathrm{sparse}}) \le n/w_d$.
    \item By the Numerically Sparse-Part Frobenius Bound (Theorem~\ref{thm:sparse_frob_bound}), its squared Frobenius mass is constrained by the numerical-sparsity threshold: $\|X_{\mathrm{sparse}}\|_F^2 \le \epsilon n w_d$.
\end{enumerate}
Applying the Rank--Frobenius to Schatten Transfer (Lemma~\ref{lem:rank_schatten_transfer}) with bounds $r = n/w_d$ and $V = \epsilon n w_d$ yields:
\begin{equation*}
\|X_{\mathrm{sparse}}\|_{S_{2/3}}^{2/3} \le (\epsilon n w_d)^{1/3} \left( \frac{n}{w_d} \right)^{2/3}.
\end{equation*}
Expanding the exponents inside the product gives:
\begin{equation*}
\|X_{\mathrm{sparse}}\|_{S_{2/3}}^{2/3} \le \epsilon^{1/3} n^{1/3} w_d^{1/3} n^{2/3} w_d^{-2/3} = \epsilon^{1/3} n w_d^{-1/3}.
\end{equation*}
This exact algebraic cancellation bounds the capacity of the numerically sparse vectors.
\end{proof}

\begin{theorem}[Diffuse-Part Mass Lower Bound]
\label{thm:diffuse_mass_lower_bound}
By setting the partition tuning parameter to exactly $\epsilon = c_0^3 / 8 = 1/96$, the diffuse projection $X_{\mathrm{diffuse}}$ satisfies the scale-wise lower bound
\begin{equation}
\|X_{\mathrm{diffuse}}\|_{S_{2/3}}^{2/3} \ge \frac{c_0}{2} n w_d^{-1/3},
\end{equation}
where $c_0 = (1/12)^{1/3}$ is the absolute constant established in Theorem~\ref{thm:Xd_schatten}.
\end{theorem}
\begin{proof}
The direct right-sided Haar projection $X_d$ decomposes as the sum of the sparse and diffuse projections: $X_d = X_{\mathrm{sparse}} + X_{\mathrm{diffuse}}$. By the $p$-triangle inequality for the Schatten $2/3$ quasi-norm (Theorem~\ref{thm:p_triangle}), the mass is subadditive for $p=2/3$:
\begin{equation*}
\|X_d\|_{S_{2/3}}^{2/3} \le \|X_{\mathrm{sparse}}\|_{S_{2/3}}^{2/3} + \|X_{\mathrm{diffuse}}\|_{S_{2/3}}^{2/3}.
\end{equation*}
From Theorem~\ref{thm:Xd_schatten}, the target projection $X_d$ contains exactly $n/w_d$ singular values of size $\Theta(w_d)$, which dictates that its total Schatten $2/3$ mass is bounded from below by:
\begin{equation*}
\|X_d\|_{S_{2/3}}^{2/3} \ge c_0 n w_d^{-1/3}.
\end{equation*}
Substituting the upper bound for the sparse-part mass from Theorem~\ref{thm:sparse_schatten_bound} into the inequality yields:
\begin{equation*}
c_0 n w_d^{-1/3} \le \epsilon^{1/3} n w_d^{-1/3} + \|X_{\mathrm{diffuse}}\|_{S_{2/3}}^{2/3}.
\end{equation*}
We now tune the partition threshold by setting $\epsilon = c_0^3 / 8$. This choice implies $\epsilon^{1/3} = c_0 / 2$. Substituting this parameter into the inequality yields:
\begin{equation*}
c_0 n w_d^{-1/3} \le \frac{c_0}{2} n w_d^{-1/3} + \|X_{\mathrm{diffuse}}\|_{S_{2/3}}^{2/3}.
\end{equation*}
Subtracting $\frac{c_0}{2} n w_d^{-1/3}$ from both sides yields the remaining lower bound:
\begin{equation*}
\|X_{\mathrm{diffuse}}\|_{S_{2/3}}^{2/3} \ge \frac{c_0}{2} n w_d^{-1/3}.
\end{equation*}
Noting that $c_0^3 = 1/12$, the explicit value of our tuning parameter is exactly $\epsilon = \frac{1}{12} \cdot \frac{1}{8} = \frac{1}{96}$.
\end{proof}

\begin{remark}
The preceding estimates show that the diffuse projection carries Schatten
$2/3$ mass at least $(c_0/2)n w_d^{-1/3}$ at every scale.  In
Section~\ref{sec:haar-projection-estimate}, we bound this contribution in terms of
the numerical sparsity of the corresponding rows of $B$.
\end{remark}

\subsection{Haar Projection Estimate for Diffuse Rows}
\label{sec:haar-projection-estimate}
In the previous subsections, we analyzed the contribution of $Q$ at each dyadic scale, partitioned the factorization factors into numerically sparse and diffuse components, and proved via the Rank--Frobenius to Schatten Transfer that the diffuse vectors must carry a macroscopic portion of the Schatten $2/3$ mass at each scale. We now bound how effectively any row of $B$ can synthesize Haar wavelets at a fixed scale in terms of its numerical sparsity, bounding its fractional capacity contribution.

\begin{lemma}[Haar Projection Estimate]
\label{lem:haar_projection_estimate}
Let $b \in \mathbb{R}^n$ be any vector with numerical sparsity $S(b) = \|b\|_1^2/\|b\|_2^2$ (with $S(b)=1$ if $b=0$). The $\ell_2$ energy of its projection $P_d b$ onto the Haar wavelet subspace at scale $d$ (with wavelet support width $w_d = 2^d$) is bounded by:
\begin{equation}
    \|P_d b\|_2^2 \le \frac{\|b\|_1^2}{\max(w_d, S(b))}.
\end{equation}
\end{lemma}
\begin{proof}
If $b=0$, the inequality holds trivially. Assume $b \neq 0$. We derive two independent bounds on the projection energy $\|P_d b\|_2^2$.

First, we bound the energy based on the wavelet support size $w_d$. The subspace at scale $d$ is spanned by the mutually orthogonal Haar wavelets $\{h_{d,j}\}_{j=1}^{n/w_d}$, each supported on a disjoint interval $I_{d,j}$ of length $w_d$. Thus, the projection energy is exactly the sum of the squared inner products:
\begin{equation}
    \|P_d b\|_2^2 = \sum_{j=1}^{n/w_d} \langle b, h_{d,j} \rangle^2.
\end{equation}
Because the non-zero entries of $h_{d,j}$ have magnitude $w_d^{-1/2}$, its $\ell_\infty$ norm is $\|h_{d,j}\|_\infty = w_d^{-1/2}$. By H\"older's inequality, we bound the magnitude of the inner product by the local $\ell_1$ norm of $b$ over the support interval $I_{d,j}$:
\begin{equation}
    |\langle b, h_{d,j} \rangle| = \left| \sum_{i \in I_{d,j}} b_i h_{d,j}(i) \right| \le \sum_{i \in I_{d,j}} |b_i| \|h_{d,j}\|_\infty = \frac{1}{\sqrt{w_d}} \sum_{i \in I_{d,j}} |b_i|.
\end{equation}
Squaring this inequality yields:
\begin{equation}
    \langle b, h_{d,j} \rangle^2 \le \frac{1}{w_d} \left( \sum_{i \in I_{d,j}} |b_i| \right)^2.
\end{equation}
Summing over all disjoint wavelets at scale $d$:
\begin{equation}
    \|P_d b\|_2^2 \le \frac{1}{w_d} \sum_{j=1}^{n/w_d} \left( \sum_{i \in I_{d,j}} |b_i| \right)^2.
\end{equation}
Since the intervals $I_{d,j}$ are mutually disjoint, we can apply the elementary inequality $\sum_j x_j^2 \le (\sum_j x_j)^2$ for non-negative sequences, by setting $x_j = \sum_{i \in I_{d,j}} |b_i|$:
\begin{equation}
    \sum_{j=1}^{n/w_d} \left( \sum_{i \in I_{d,j}} |b_i| \right)^2 \le \left( \sum_{j=1}^{n/w_d} \sum_{i \in I_{d,j}} |b_i| \right)^2 \le \left( \sum_{i=1}^n |b_i| \right)^2 = \|b\|_1^2.
\end{equation}
Therefore, we obtain the spatial bound:
\begin{equation}
\label{eq:bound_wd}
    \|P_d b\|_2^2 \le \frac{\|b\|_1^2}{w_d}.
\end{equation}

Second, we bound the energy based on the numerical sparsity $S(b)$. Since $P_d$ is an orthogonal projection, it cannot increase the total $\ell_2$ energy of $b$. Thus, we have the immediate spectral bound:
\begin{equation}
\label{eq:bound_S}
    \|P_d b\|_2^2 \le \|b\|_2^2 = \frac{\|b\|_1^2}{S(b)}.
\end{equation}

Since $\|P_d b\|_2^2$ must simultaneously satisfy both the spectral bound \eqref{eq:bound_S} and the spatial bound \eqref{eq:bound_wd}, we conclude:
\begin{equation}
    \|P_d b\|_2^2 \le \min\left( \frac{\|b\|_1^2}{S(b)}, \frac{\|b\|_1^2}{w_d} \right) = \frac{\|b\|_1^2}{\max(w_d, S(b))}.
\end{equation}
\end{proof}

\begin{corollary}[Fractional Capacity Bound for a Single Row]
\label{cor:row_projection}
Applying this estimate to the rows of the factorization matrix $B$, let $b_k = (B_{k,:})^T$ be the $k$-th row of $B$, with $\ell_1$ norm $L_k = \|b_k\|_1$ and numerical sparsity $S_k = L_k^2 / \|b_k\|_2^2$. The fractional $\ell_2^{2/3}$ capacity provided by this row at scale $d$ satisfies:
\begin{equation}
    \|P_d b_k\|_2^{2/3} \le L_k^{2/3} \max(w_d, S_k)^{-1/3} = L_k^{2/3} w_d^{-1/3} \min\left(1, \left(\frac{w_d}{S_k}\right)^{1/3}\right).
\end{equation}
\end{corollary}

We now evaluate the total fractional capacity that the diffuse components can contribute to the scale projection $X_d = Q P_d$. The index set $\mathcal{I}_{\mathrm{diffuse}}(d)$ consists of the diffuse vectors at scale $d$, namely those with numerical sparsity $S_k > \epsilon w_d / \Delta$ (where $\Delta = M_A^2 M_B^2$ and $\epsilon = c_0^3 / 8$, as established in Section~\ref{sec:schatten-transfer}). The diffuse part of the projection is $X_{\mathrm{diffuse}} = \sum_{k \in \mathcal{I}_{\mathrm{diffuse}}(d)} A_{:, k} (P_d b_k)^T$.

\begin{theorem}[Diffuse-Row Fractional Cost Bound]
\label{thm:diffuse_fractional_bound}
The Schatten $2/3$ mass of the diffuse projection at scale $d$ is bounded from above by the numerical sparsity of its constituent vectors:
\begin{equation}
\|X_{\mathrm{diffuse}}\|_{S_{2/3}}^{2/3} \le w_d^{-1/3} \sum_{k \in \mathcal{I}_{\mathrm{diffuse}}(d)} N_k^{1/3} L_k^{2/3} \min\left(1, \left(\frac{w_d}{S_k}\right)^{1/3}\right).
\end{equation}
\end{theorem}
\begin{proof}
As established in Theorem~\ref{thm:p_triangle}, the Schatten $2/3$ quasi-norm satisfies the $p$-triangle inequality for $p=2/3$. Applying this to the rank-1 decomposition of $X_{\mathrm{diffuse}}$, we bound the quasi-norm of the sum by the sum of the quasi-norms of its components:
\begin{equation}
\|X_{\mathrm{diffuse}}\|_{S_{2/3}}^{2/3} = \left\| \sum_{k \in \mathcal{I}_{\mathrm{diffuse}}(d)} A_{:, k} (P_d b_k)^T \right\|_{S_{2/3}}^{2/3} \le \sum_{k \in \mathcal{I}_{\mathrm{diffuse}}(d)} \left\| A_{:, k} (P_d b_k)^T \right\|_{S_{2/3}}^{2/3}.
\end{equation}
Because each matrix $A_{:, k} (P_d b_k)^T$ has rank at most $1$, its sole non-zero singular value is the product of the $\ell_2$ norms of its defining vectors. Thus, its Schatten $2/3$ mass is exactly:
\begin{equation}
\left\| A_{:, k} (P_d b_k)^T \right\|_{S_{2/3}}^{2/3} = \|A_{:, k}\|_2^{2/3} \|P_d b_k\|_2^{2/3} = N_k^{1/3} \|P_d b_k\|_2^{2/3}.
\end{equation}
Substituting the bound on $\|P_d b_k\|_2^{2/3}$ from Corollary \ref{cor:row_projection} yields:
\begin{equation}
\left\| A_{:, k} (P_d b_k)^T \right\|_{S_{2/3}}^{2/3} \le N_k^{1/3} L_k^{2/3} w_d^{-1/3} \min\left(1, \left(\frac{w_d}{S_k}\right)^{1/3}\right).
\end{equation}
Summing this upper bound over all $k \in \mathcal{I}_{\mathrm{diffuse}}(d)$ proves the theorem.
\end{proof}

\begin{theorem}[Scale-Wise Diffuse-Row Cost Bound]
\label{thm:scale_wise_diffuse_bound}
For any scale $d \in \{1, \dots, \log_2 n\}$ with spatial width $w_d$, synthesizing the diffuse mass of the target projection $X_d$ requires the fractional capacity of the diffuse factorization vectors to satisfy:
\begin{equation}
\sum_{k \in \mathcal{I}_{\mathrm{diffuse}}(d)} N_k^{1/3} L_k^{2/3} \min\left(1, \left(\frac{w_d}{S_k}\right)^{1/3}\right) \ge \frac{c_0}{2} n,
\end{equation}
where $c_0 = (1/12)^{1/3}$ is the universal absolute constant established in Theorem~\ref{thm:Xd_schatten}.
\end{theorem}
\begin{proof}
By the Rank--Frobenius to Schatten Transfer (Section~\ref{sec:schatten-transfer}), we proved that the numerically sparse vectors cannot carry the full required mass, forcing the diffuse component to satisfy:
\begin{equation}
\|X_{\mathrm{diffuse}}\|_{S_{2/3}}^{2/3} \ge \frac{c_0}{2} n w_d^{-1/3}.
\end{equation}
Combining this strict geometric lower bound with the factorization upper bound from Theorem \ref{thm:diffuse_fractional_bound}, we obtain:
\begin{equation}
\frac{c_0}{2} n w_d^{-1/3} \le w_d^{-1/3} \sum_{k \in \mathcal{I}_{\mathrm{diffuse}}(d)} N_k^{1/3} L_k^{2/3} \min\left(1, \left(\frac{w_d}{S_k}\right)^{1/3}\right).
\end{equation}
Because the spatial support width satisfies $w_d \ge 2 > 0$, we can cancel the common $w_d^{-1/3}$ factor from both sides, yielding the final scale-wise inequality.
\end{proof}

\begin{remark}
For a fixed row $k$, the scale-$d$ multiplier equals $1$ when $w_d\ge S_k$
and equals $(w_d/S_k)^{1/3}$ when $w_d<S_k$.  The row appears in the
diffuse sum only when $w_d<S_k\Delta/\epsilon$; otherwise it belongs to the
numerically sparse part controlled in Sections~\ref{sec:trace-bottleneck}
and~\ref{sec:schatten-transfer}.  These facts bound the total contribution of
each row across scales, as shown next.
\end{remark}

\subsection{Log-Log Aggregation and Final Lower Bound}
\label{sec:aggregation}
In Section~\ref{sec:haar-projection-estimate}, we established the Scale-Wise Diffuse-Row Cost Bound (Theorem~\ref{thm:scale_wise_diffuse_bound}), which demonstrates that synthesizing the Haar wavelets at any single spatial scale $d$ imposes a strict macroscopic lower bound on the fractional capacity of the diffuse vectors. We now aggregate this capacity demand across all $\log_2 n$ dyadic scales. The pivotal insight of this aggregation is that the dynamic partition established in Section~\ref{sec:trace-bottleneck} truncates the required summation for any given vector. A factor only contributes to the diffuse subset at scales where its numerical sparsity is sufficiently large ($w_d < S_k \Delta / \epsilon$). This limits the number of coarse scales (where $w_d \ge S_k$) at which the vector can contribute. This geometric truncation forces a lower bound on the aggregate fractional rank-one cost, averting the need for heavy component deflation and closing the factorization barrier.

\begin{theorem}[Per-Row Scale Multiplier Bound]
\label{thm:per_row_scale_multiplier}
For any vector $k \in \{1, \dots, m\}$ with numerical sparsity $S_k$, define the active scales as the subset of dyadic scales $d \in \{1, \dots, \log_2 n\}$ where $k$ is classified as diffuse:
\begin{equation}
\mathcal{D}_k = \left\{ d \in \{1, \dots, \log_2 n\} \;\middle|\; w_d < \frac{S_k \Delta}{\epsilon} \right\}.
\end{equation}
Assume the factorization capacity satisfies $\Delta = M_A^2 M_B^2 \le (\log_2 n)^3$. Then the sum of the fractional capacity multipliers over the active scales is bounded by:
\begin{equation}
\Sigma_k = \sum_{d \in \mathcal{D}_k} \min\left(1, \left(\frac{w_d}{S_k}\right)^{1/3}\right) \le 3 \log_2 \log_2 n + 13.
\end{equation}
\end{theorem}
\begin{proof}
Throughout this proof we use the fixed tuning parameter $\epsilon=1/96$ from Theorem~\ref{thm:diffuse_mass_lower_bound}. By Lemma~\ref{lem:nontrivial_cost}, $\Delta\ge 1$, and hence $\Delta/\epsilon\ge 96>1$.
We split the sum into two disjoint sets of scales: coarse scales and fine scales relative to the numerical sparsity $S_k$.

\noindent \textbf{1. Coarse Scales:} Let $\mathcal{D}_k^{\text{coarse}} = \{ d \in \mathcal{D}_k \mid w_d \ge S_k \}$.
For these scales, $w_d / S_k \ge 1$, so the minimum evaluates to $1$. The number of such scales is exactly the size of this set. Since $d$ is an integer and $w_d = 2^d$, the condition translates to:
\begin{equation}
S_k \le 2^d < \frac{S_k \Delta}{\epsilon}.
\end{equation}
Taking the base-2 logarithm, we obtain:
\begin{equation}
\log_2 S_k \le d < \log_2 S_k + \log_2\left(\frac{\Delta}{\epsilon}\right).
\end{equation}
The number of integers in any half-open interval of length $L$ is bounded by $\lceil L \rceil \le L + 1$. Thus, the number of coarse scales is at most $\log_2(\Delta/\epsilon) + 1$. Under the assumption that $\Delta \le (\log_2 n)^3$, we have:
\begin{equation}
\log_2\left(\frac{\Delta}{\epsilon}\right) + 1 \le \log_2\left((\log_2 n)^3\right) + \log_2\left(\frac{1}{\epsilon}\right) + 1 = 3 \log_2 \log_2 n + \log_2\left(\frac{1}{\epsilon}\right) + 1.
\end{equation}
Recall from Theorem~\ref{thm:diffuse_mass_lower_bound} that $\epsilon = 1/96$. Thus $\log_2(1/\epsilon) = \log_2 96 \approx 6.58 < 7$. Therefore, the sum over coarse scales is bounded by $3 \log_2 \log_2 n + 8$.

\noindent \textbf{2. Fine Scales:} Let $\mathcal{D}_k^{\text{fine}} = \{ d \in \mathcal{D}_k \mid w_d < S_k \}$.
For these scales, $w_d / S_k < 1$, and the minimum evaluates to $(w_d / S_k)^{1/3} = 2^{(d - \log_2 S_k)/3}$. Let $d_{\max} = \lfloor \log_2 S_k \rfloor$. All fine scales satisfy $d \le d_{\max}$. We can extend the sum down to $d \to -\infty$ to obtain a strict upper bound via an infinite geometric series:
\begin{equation}
\sum_{d \in \mathcal{D}_k^{\text{fine}}} \left(\frac{2^d}{S_k}\right)^{1/3} \le \sum_{j=0}^{\infty} \left( \frac{2^{d_{\max} - j}}{S_k} \right)^{1/3}.
\end{equation}
Since $2^{d_{\max}} \le S_k$, the leading term is at most $1$. Factoring this out, we get:
\begin{equation}
\sum_{j=0}^{\infty} (2^{-1/3})^j = \frac{1}{1 - 2^{-1/3}} \approx 4.85 < 5.
\end{equation}
Adding the bounds from both regimes, the total capacity multiplier is bounded by:
\begin{equation}
\Sigma_k \le 3 \log_2 \log_2 n + 8 + 5 = 3 \log_2 \log_2 n + 13.
\end{equation}
This completes the proof.
\end{proof}

\begin{theorem}[Aggregate Fractional Rank-One Lower Bound]
\label{thm:aggregate_fractional_bound}
If the factorization capacity satisfies $\Delta = M_A^2 M_B^2 \le (\log_2 n)^3$, then the fractional rank-one cost is bounded below by:
\begin{equation}
\sum_{k=1}^m N_k^{1/3} L_k^{2/3} \ge \Omega\left( n \frac{\log_2 n}{\log_2 \log_2 n} \right).
\end{equation}
\end{theorem}
\begin{proof}
We sum the Scale-Wise Diffuse-Row Cost Bound (Theorem~\ref{thm:scale_wise_diffuse_bound}) over all spatial scales $d \in \{1, \dots, \log_2 n\}$. By definition, the index set $\mathcal{I}_{\mathrm{diffuse}}(d)$ contains exactly those vectors $k$ where $S_k > \epsilon w_d / \Delta$, or equivalently, $d \in \mathcal{D}_k$. Summing over all $\log_2 n$ scales yields:
\begin{equation}
\sum_{d=1}^{\log_2 n} \sum_{k \in \mathcal{I}_{\mathrm{diffuse}}(d)} N_k^{1/3} L_k^{2/3} \min\left(1, \left(\frac{w_d}{S_k}\right)^{1/3}\right) \ge \sum_{d=1}^{\log_2 n} \frac{c_0}{2} n = \frac{c_0}{2} n \log_2 n.
\end{equation}
Because all terms are non-negative, we can exchange the order of summation. Each vector $k$ is summed over exactly its active scales $\mathcal{D}_k$:
\begin{equation}
\sum_{k=1}^m N_k^{1/3} L_k^{2/3} \sum_{d \in \mathcal{D}_k} \min\left(1, \left(\frac{w_d}{S_k}\right)^{1/3}\right) = \sum_{k=1}^m N_k^{1/3} L_k^{2/3} \Sigma_k \ge \frac{c_0}{2} n \log_2 n.
\end{equation}
Applying Theorem~\ref{thm:per_row_scale_multiplier} to bound the capacity multiplier $\Sigma_k$:
\begin{equation}
\sum_{k=1}^m N_k^{1/3} L_k^{2/3} (3 \log_2 \log_2 n + 13) \ge \frac{c_0}{2} n \log_2 n.
\end{equation}
Dividing both sides by $(3 \log_2 \log_2 n + 13)$, we obtain:
\begin{equation}
\sum_{k=1}^m N_k^{1/3} L_k^{2/3} \ge \frac{c_0 n \log_2 n}{6 \log_2 \log_2 n + 26}.
\end{equation}
For sufficiently large $n$, the denominator is positive and scales as $\Theta(\log \log n)$, yielding:
\begin{equation}
\sum_{k=1}^m N_k^{1/3} L_k^{2/3} \ge \Omega\left( n \frac{\log_2 n}{\log_2 \log_2 n} \right).
\end{equation}
This concludes the proof.
\end{proof}

\begin{restatedmaintheorem}[Main theorem]
There is an absolute constant $c>0$ such that, for all sufficiently large $n$,
all $m\in\mathbb{N}$, and all real factorizations
$Q=AB$ with $A\in\mathbb{R}^{n\times m}$ and
$B\in\mathbb{R}^{m\times n}$,
\[
\|A\|_{2\to\infty}\|B\|_{1\to1}
\ge
c\frac{\log^{3/2}n}{(\log\log n)^{3/2}}.
\]
Consequently,
\[
\gamma_{2,1}(Q)
=
\Omega\!\left(
\frac{\log^{3/2}n}{(\log\log n)^{3/2}}
\right).
\]
\end{restatedmaintheorem}

\begin{proof}[Proof of Theorem~\ref{thm:main_result}]
Let $\Delta = M_A^2 M_B^2$. We proceed by case analysis on $\Delta$.

\noindent \textbf{Case 1:} $\Delta > (\log_2 n)^3$.
Then $M_A M_B = \Delta^{1/2} > (\log_2 n)^{1.5}$. Since $(\log_2 n)^{1.5}$ dominates $\Omega\left(\frac{\log^{1.5} n}{(\log \log n)^{1.5}}\right)$ asymptotically, the bound holds trivially.

\noindent \textbf{Case 2:} $\Delta \le (\log_2 n)^3$.
By Theorem~\ref{thm:aggregate_fractional_bound}, the fractional rank-one cost satisfies:
\begin{equation}
\sum_{k=1}^m N_k^{1/3} L_k^{2/3} \ge \Omega\left( n \frac{\log n}{\log \log n} \right).
\end{equation}
This exactly satisfies the Fractional Rank-One Bound (Equation~\ref{eq:fractional_cost_target}) established in Section~\ref{sec:fractional-reduction}. Therefore, by the Fractional H\"older Reduction (Theorem~\ref{thm:fractional_holder}), the factorization norm is bounded by:
\begin{equation}
M_A M_B \ge \Omega\left(\frac{\log^{1.5} n}{(\log \log n)^{1.5}}\right).
\end{equation}
Taking the infimum over all real factorizations gives the claimed bound for
$\gamma_{2,1}(Q)$.
\end{proof}

\bibliographystyle{alpha}
\bibliography{matrix_arxiv}

\end{document}